\documentclass[orivec]{llncs}

\usepackage{etex}
\usepackage{marginnote}
\usepackage{amsmath, amssymb, amsfonts}
\usepackage{paralist}
\usepackage{alltt}
\usepackage{wrapfig}
\usepackage{stmaryrd}
\usepackage{MnSymbol}
\usepackage{paralist}
\usepackage{booktabs}

\usepackage{thmtools,thm-restate}
\usepackage[usenames,dvipsnames]{xcolor}
\usepackage{caption}
\usepackage{subcaption}
\usepackage{floatrow}
\usepackage{cite}
\usepackage{calc}

\usepackage{algorithm}
\usepackage[noend]{algpseudocode}

\usepackage{tikz}
\usetikzlibrary{matrix}
\usetikzlibrary{fit}
\usetikzlibrary{automata}
\usetikzlibrary{shapes}

\newcommand\comment[1]{}

\usepackage{xspace}
\usepackage{array}
\usepackage{proof}

\newcommand{\TO}{$\infty$} 

\newcolumntype{t}{>{\tt}l}

\newif\ifappendix
\appendixtrue

\newcommand{\hbra}{
  \hbox to \columnwidth{\vrule width0.3mm height 1.8mm depth-0.3mm
    \leaders\hrule height1.8mm depth-1.5mm\hfill
    \vrule width0.3mm height 1.8mm depth-0.3mm}}
\newcommand{\hket}{
  \hbox to \columnwidth{\vrule width0.3mm height1.5mm
    \leaders\hrule height0.3mm\hfill
    \vrule width0.3mm height1.5mm}}

\makeatletter
  \newcommand{\addToLabel}[1]{%
    \protected@edef\@currentlabel{\@currentlabel#1}%
  }
\makeatother
\newcommand{\ratio}{.35}

\newcounter{excounter}

\newenvironment{display}[2][\ratio]{\vspace{-0.5ex}
\begin{tabbing}
     \=  \=  \= \kill
    \textbf{#2}\\[-.8ex]
    \hbra\\[-.8ex]
  }{\\[-.8ex]\hket
  \end{tabbing}\vspace{-1ex}}

\newcounter{rule}

\newcommand{\staterule}[4][]{%
  \refstepcounter{rule}%
  \addToLabel{(#2)}
  $\begin{array}[b]{@{}l@{}}%
    \begin{array}{@{}c@{}}
      #3\\
      \hline
      \raisebox{0ex}[2.5ex]{\strut}#4%
    \end{array}
    \mbox{{\sc #2} #1}\\%
  \end{array}$}

\newcommand{\GAP}{1ex}

\newcommand{\smallstepx}[3]{\langle #1 \rangle \xrightarrow{#3} \langle #2 \rangle}
\newcommand{\smallstepxs}[3]{\langle #1 \rangle \xrightarrow{#3}\negthickspace^* \langle #2 \rangle}
\newcommand{\smallstep}[2]{\smallstepx{#1}{#2}{\epsilon}}

\newcommand{\smallstepo}[3]{\smallstepx{\VariableValuation,{#1}}{\VariableValuation,{#2}}{#3}}
\newcommand{\smallstepi}[4]{\smallstepx{\VariableValuation,{#1}}{\VariableValuation[#4],
\Skip}{#3}}

\newcommand{\kw}[1]{\ensuremath{\mathsf{#1}}}

\newcommand{\e}{e}

\newcommand{\Skip}{\kw{skip}}
\newcommand{\Seq}[2]{#1;#2}
\newcommand{\Asgn}[2]{#1:=#2}
\newcommand{\eIf}[3]{\kw{if}\ #1\ \kw{then}\ #2\ \kw{else}\ #3}
\newcommand{\eWhile}[2]{\kw{while}(#1)\ #2}

\newcommand{\Lock}[1]{\kw{lock}(#1)}

\newcommand{\Unlock}[1]{\kw{unlock}(#1)}
\newcommand{\Input}[1]{\kw{input}(#1)}
\newcommand{\Output}[2]{\kw{output}(#1,#2)}
\newcommand{\Yield}{\kw{yield}}
\newcommand{\Havoc}{\kw{havoc}}
\newcommand{\Await}[1]{\kw{await}(#1)}
\newcommand{\Signal}[1]{\kw{signal}(#1)}
\newcommand{\Reseta}[1]{\kw{reset}(#1)}

\newcommand{\VariableValuation}{\calV}
\newcommand{\evalExpression}[1]{#1[v / \VariableValuation[v]]}
\newcommand{\reads}[1]{ \mathit{Reads}(#1) }

\newcommand{\defref}[1]{Definition~\ref{def:#1}}
\newcommand{\secref}[1]{Sec.~\ref{sec:#1}}
\newcommand{\subsecref}[1]{Sec.~\ref{subsec:#1}}
\newcommand{\figref}[1]{Fig.~\ref{fig:#1}}

\newcommand{\algoref}[1]{Algo.~\ref{algo:#1}}
\newcommand{\thmref}[1]{Thm.~\ref{thm:#1}}
\newcommand{\propref}[1]{Prop.~\ref{prop:#1}}

\newcommand{\lineref}[1]{Line~\ref{line:#1}}

\newcommand\calA{\mathcal{A}}
\newcommand\calC{\mathcal{C}}

\newcommand\calV{\mathcal{V}}

\newcommand\tid{\mathit{tid}}

\newcommand\techReport{false} 

\newcommand\ifTechReport[1]{\ifthenelse{\equal{\techReport}{true}}{#1}{}}
\newcommand\ifPaper[1]{\ifthenelse{\equal{\techReport}{true}}{}{#1}}

\newcommand{\lang}[1]{\ensuremath{\mathcal L(\mbox{#1})}}
\newcommand{\Lang}[1]{\ensuremath{\mathcal L(#1)}}

\newcommand{\true}{\mathit{true}}
\newcommand{\false}{\mathit{false}}

\newcommand{\trues}{\mathsf{true}}
\newcommand{\falses}{\mathsf{false}}

\newcommand\Prog{\ensuremath{\mathit{P}}}
\newcommand\ProgState{\ensuremath{S}}

\newcommand\natset{\ensuremath{\mathbb{N}}}

\newcommand\antichain{\mathit{antichain}}
\newcommand\frontier{\mathit{frontier}}
\newcommand\overflow{\mathit{overflow}}

\newcommand\dirty{{\em dirty}}

\newcommand{\loc}{l}
\newcommand{\wlang}{\mathcal W}

\newcommand{\nfa}{NFA}
\newcommand{\nfas}{NFAs}
\newcommand{\A}{A}
\newcommand{\B}{B}

\newcommand{\early}{\eta_1}
\newcommand{\late}{\eta_2}

\newcommand{\upto}{modulo}
\newcommand{\suc}{succ}
\newcommand{\CI}{\mathrm{Clo}_I}
\newcommand{\CIk}{\mathrm{Clo}_{k,I}}
\newcommand{\assign}{\mathtt{:=}}

\newcommand{\ourtool}{{\sc Liss}}

\newcommand{\autC}{S_{\autCB}}
\newcommand{\autCB}{\B_{k,I}}

\newcommand{\thread}{\mathtt{T}}
\newcommand{\cProg}{\calC}

\newcommand{\sem}[1]{[\![ #1 ]\!]}
\newcommand{\Aut}{\calA}

\newcommand{\cex}{cex}

\newcommand{\mtt}[1]{\mathtt{#1}}

\newcommand{\mixcons}{\psi}

\newcommand\nhood{nhood}

\usepackage[plainpages = false, pdfpagelabels, pdfborder={0 0 0}, pdfdisplaydoctitle=true,pdflang=en-gb,pdfencoding=auto,bookmarksnumbered=true,bookmarks=true,bookmarksopen=true,pdfstartview={XYZ null null 1}]{hyperref}

\begin{document}
\title{From Non-preemptive to Preemptive Scheduling using
  Synchronization Synthesis
\thanks{This research was supported in part by the European Research
  Council (ERC) under grant 267989 (QUAREM), by the Austrian Science
  Fund (FWF) under grants S11402-N23 (RiSE) and Z211-N23 (Wittgenstein
  Award), by NSF under award CCF 1421752 and the Expeditions award CCF
  1138996, by the Simons Foundation, and by a gift from the Intel Corporation.} 
} 
\author{Pavol {\v C}ern{\'y}\inst{1} \and Edmund M. Clarke\inst{2} \and Thomas A. Henzinger\inst{3}
\and Arjun Radhakrishna\inst{4} \and Leonid
Ryzhyk\inst{2}\and Roopsha Samanta\inst{3}\and Thorsten Tarrach\inst{3}}
\institute{University of Colorado Boulder \and Carnegie Mellon University \and IST Austria 
\and University of Pennsylvania}
\maketitle

\begin{abstract}
We present a computer-aided programming approach to concurrency. 
The approach allows programmers to program assuming a friendly,
non-preemptive scheduler, and our synthesis procedure inserts
synchronization to ensure that the final program works even with a 
preemptive scheduler.
The correctness specification is implicit, inferred from
the non-preemptive behavior. 
Let us consider sequences of calls that the program makes to an
external interface.  
The specification requires that any such sequence produced under a preemptive 
scheduler should be included in the set of such sequences produced
under a non-preemptive scheduler. 
The solution is based on a finitary abstraction, 
an algorithm for bounded language inclusion \upto~an
independence relation, and rules for inserting
synchronization.  
We apply the approach to device-driver programming,
where the driver threads call the software interface of the device 
and the API provided by the operating system.   
Our experiments demonstrate that our synthesis method is precise 
and efficient, and, since it does not require explicit 
specifications, is more practical than the conventional approach 
based on user-provided assertions.

\end{abstract}

\section{Introduction}
\label{sec:intro}

Concurrent shared-memory programming is notoriously difficult and error-prone.  
Program synthesis for concurrency aims to mitigate this complexity
by synthesizing synchronization code automatically~\cite{clarke1982design,CAV13,CAV14,POPL15}.  
However, 
specifying the programmer's intent may be a challenge in itself.
Declarative mechanisms, such as assertions, suffer from
the drawback that it is difficult to ensure that the
specification is complete 
and fully captures the programmer's intent. 

We propose a solution where the specification is {\em implicit}. 
We observe that a core difficulty in concurrent programming originates 
from the fact that the scheduler can \emph{preempt} the execution 
of a thread at any time.  We therefore give the developer the option to program 
assuming a friendly, \emph{non-preemptive}, scheduler.  Our tool 
automatically synthesizes synchronization code to ensure that every
behavior of the program under preemptive scheduling is included in the
set of behaviors produced under non-preemptive scheduling. 
Thus, we use the non-preemptive semantics as an implicit 
correctness specification.

The non-preemptive scheduling model 
dramatically simplifies the development of concurrent software, including operating system (OS) kernels, 
network servers, database systems, etc.~\cite{Sadowski10,Ryzhyk_CKH_09}.  In this model, 
a thread can only be descheduled by voluntarily yielding control, e.g., by 
invoking a blocking operation.  Synchronization primitives may be used 
for communication between threads, e.g., a producer thread may use a semaphore 
to notify the consumer about availability of data. However, one does not need to worry 
about protecting accesses to shared state: a series of memory accesses 
executes atomically as long as the scheduled thread does not yield.

In defining behavioral equivalence between preemptive and non-preemptive 
executions, we focus on externally observable program behaviors: 
two program executions are \emph{observationally equivalent} if they generate 
the same sequences of calls to interfaces of interest.
This approach facilitates modular synthesis where a module's behavior 
is characterized in terms of its interaction with other modules. 
Given a multi-threaded program $\cProg$ and a synthesized program
$\cProg'$ obtained by adding synchronization  
to $\cProg$, $\cProg'$ is {\em preemption-safe} w.r.t. $\cProg$ if for each 
execution of $\cProg'$ under a preemptive scheduler, there is 
an observationally equivalent non-preemptive execution of $\cProg$. 
Our synthesis goal is to automatically generate a 
preemption-safe version of the input program.

We rely on abstraction to achieve efficient synthesis 
of multi-threaded programs.  We propose a simple, {\em data-oblivious}
abstraction inspired by an analysis of synchronization 
patterns in OS code, which tend to be independent of data values.  
The abstraction tracks types of accesses (read or write) to each
memory location while ignoring their values.
In addition, the abstraction tracks branching choices.
Calls to an external interface are modeled as writes to a special 
memory location, with independent interfaces modeled as separate 
locations. To the best of our knowledge, our proposed abstraction 
is yet to be explored in the verification and synthesis literature. 


Two abstract program executions are observationally 
equivalent if they are equal \upto~the classical independence relation 
$I$ on memory accesses: accesses to different
locations are independent, and accesses to the same location are independent 
iff they are both read accesses. Using this notion of equivalence, the
notion of preemption-safety is extended to abstract programs.

Under abstraction, we model each thread as a nondeterministic finite
automaton (\nfa) 
over a finite alphabet, with each symbol corresponding to a 
read or a write to a particular variable. This enables us to construct 
\nfas\ $N$, representing the abstraction of the original program $\cProg$ 
under non-premptive scheduling, and $P$,
representing the abstraction of the synthesized program $\cProg'$ under preemptive scheduling.
We show that preemption-safety of $\cProg'$ w.r.t. $\cProg$ is implied by 
preemption-safety of the abstract synthesized program w.r.t. the abstract
original program, which, in turn,  
is implied by language inclusion \upto~$I$ of \nfas~$P$ and $N$. 
While the problem of language inclusion \upto~an independence relation  
is undecidable~\cite{bertoni1982equivalence}, we show that the
antichain-based algorithm for standard language  
inclusion~\cite{de2006antichains} can be adapted to decide a bounded version of language
inclusion \upto~an independence relation. 

Our overall synthesis procedure works as follows: we run the algorithm for
bounded language inclusion \upto~$I$, iteratively increasing the bound, until
it reports that the inclusion holds, or finds a counterexample, or reaches a
timeout.  In the first case, the synthesis procedure terminates successfully.
In the second case, the counterexample is generalized 
to a set of counterexamples represented as a Boolean combination of ordering constraints
over control-flow locations (as in \cite{POPL15}). These constraints are analyzed for
patterns indicating the type of concurrency bug (atomicity,
ordering violation) and the type of applicable fix (lock insertion, statement reordering). 
After applying the fix(es), the procedure is restarted from scratch;
the process continues until we find a preemption-safe program, or reach a timeout.

We implemented our synthesis procedure in a new prototype tool called 
\ourtool{} (Language Inclusion-based Synchronization Synthesis) and
evaluated it on a series of device driver  
benchmarks, including an Ethernet driver for Linux and the 
synchronization skeleton of a USB-to-serial controller driver.
First, \ourtool\ was able to detect and eliminate 
all but two known race conditions in our examples;
these included one race condition that we 
previously missed when synthesizing from explicit 
specifications \cite{CAV14}, due to a missing assertion.  Second, our 
abstraction proved highly efficient: \ourtool{} runs an order of 
magnitude faster on the more complicated examples than our 
previous synthesis tool based on the CBMC  
model checker.  Third, our coarse abstraction proved 
surprisingly precise in practice: across all our benchmarks, 
we only encountered three program locations  
where manual abstraction refinement was needed to avoid 
the generation of unnecessary synchronization.
Overall, our evaluation strongly supports the use of
the implicit specification approach based on non-preemptive 
scheduling semantics as well as the use of the data-oblivious
abstraction to achieve practical synthesis for real-world systems
code.  

\noindent {\bf Contributions}. 
First, we propose a new specification-free approach to synchronization
synthesis.  
Given a program written assuming a friendly, non-preemptive scheduler, we 
automatically generate a preemption-safe version of the program.
Second, we introduce a novel abstraction scheme and use it to reduce
preemption-safety to language inclusion \upto~an independence relation.
Third, we present the first language inclusion-based 
synchronization synthesis procedure and tool for concurrent 
programs. 
Our synthesis procedure includes a 
new algorithm for a bounded version of our inherently undecidable
language inclusion problem.
Finally, we evaluate our 
synthesis procedure on several examples. To the best of our 
knowledge, \ourtool{} is the first synthesis tool capable of 
handling realistic (albeit simplified) device driver code, while 
previous tools were evaluated on small fragments of driver code or 
on manually extracted synchronization skeletons.
%

\noindent{\bf Related work}.
Synthesis of synchronization is an active research
area~\cite{bloem,Vechev:2010:ASS:1706299.1706338,Cherem:2008:ILA:1375581.1375619,ramalingam,SLJB08,shanlu,CAV13,CAV14,POPL15}.    
Closest to our work is a recent paper by Bloem et al.~\cite{bloem},
which uses implicit specifications for
synchronization synthesis. 
While their specification is given
by sequential behaviors, ours is given by non-preemptive
behaviors. This makes our approach applicable to scenarios where threads
need to communicate explicitly.
Further, correctness in \cite{bloem} is determined by comparing values at the end
of the execution. In contrast, we compare sequences of events, which
serves as a more 
suitable specification for infinitely-looping reactive systems. 

Many efforts in synthesis of synchronization focus on user-provided
specifications, such as assertions (our previous work~\cite{CAV13,CAV14,POPL15}).
However, it is hard to determine if a given set of assertions
represents a complete specification.
In this paper, we are solving language inclusion, a 
computationally harder problem than
reachability. However, due to our abstraction, our tool performs
significantly better than tools from~\cite{CAV13,CAV14}, which are based
on a mature model checker (CBMC~\cite{cbmc}). 
Our abstraction is reminiscent of previously used abstractions
that track reads and writes to individual locations (e.g.,~\cite{VYRS10,AKNP14}).
However, our abstraction is novel as it additionally tracks some
control-flow information (specifically, the branches taken) 
giving us higher precision with almost negligible computational cost.
The synthesis part of our approach is based on~\cite{POPL15}.

In~\cite{Vechev:2010:ASS:1706299.1706338} the authors rely on
assertions for synchronization synthesis and
include iterative abstraction refinement in their framework. This is
an interesting extension to pursue for our abstraction.
In other related work, CFix \cite{shanlu} can detect and fix
concurrency bugs by identifying simple bug patterns in the code.

\vspace{-1ex}
\section{Illustrative Example}
\label{sec:runningexample}

\begin{figure}[tb]
\begin{minipage}{0.22\textwidth}%
\scriptsize{
\begin{alltt}
void open\_dev() \{
1: while (*) \{  
2:   if (open==0) \{ 
3:     power\_up();   
4:   \}   
5:   open=open+1;
6:  yield; \} \}          
\end{alltt}
}          
\end{minipage}
~
\begin{minipage}{0.22\textwidth}%
\scriptsize{
\begin{alltt}
void close\_dev() \{  
7:  while (*) \{   
8:  if (open>0) \{            
9:    open=open-1;            
10:   if (open==0) \{         
11:    power\_down();        
12:  \} \}                        
13: yield; \} \}                            
\end{alltt}
}                  
\end{minipage}
\vrule
~
\begin{minipage}{0.22\textwidth}%
\scriptsize{
\begin{alltt}
void open\_dev\_abs() \{ 
1: while (*) \{ 
2: (A) r open; 
   if (*) \{ 
3:   (B) w dev; 
4:  \} 
5: (C) r open; 
   (D) w open; 
6: yield; \} \} 
\end{alltt}
}            
\end{minipage}
~
\begin{minipage}{0.22\textwidth}%
\scriptsize{
\begin{alltt}
void close\_dev\_abs() \{ 
7: while (*) \{ 
8:  (E) r open; 
    if (*) \{ 
9:    (F) r open; 
      (G) w open; 
10:   (H) r open; 
      if (*) \{ 
11:     (I) w dev; 
12: \} \} 
13: yield; \} \}
\end{alltt}
}            
\end{minipage}\\
(a)~~~~~~~~~~~~~~~~~~~~~~~~~~~~~~~~~~~~~~~~~~~(b)\\
\vspace{-2ex}
\caption{Running example and its abstraction}
\label{fig:example1}
\end{figure}

\autoref{fig:example1}a contains our running example. 
Consider the case where the procedures {\tt open\_dev()} and
{\tt close\_dev()} are invoked in parallel, possibly multiple times 
(modeled as a non-deterministic while loop).
The functions {\tt power\_up()} and {\tt power\_down()} represent calls to
a device.
For the non-preemptive scheduler, the sequence of calls to
the device will always be a repeating sequence of one call to {\tt power\_up()},
followed by one call to {\tt power\_down()}.
Without additional synchronization, however, 
there could be two calls to {\tt power\_up()} in a row when executing it with
a preemptive scheduler. Such a sequence is not  
observationally equivalent to any sequence that can be produced when
executing with a non-preemptive scheduler.  

\autoref{fig:example1}b contains the abstracted 
versions (we omit tracking of branching choices in the example)
of the two
procedures, {\tt open\_dev\_abs()} and
{\tt close\_dev\_abs()}. For instance, the instruction {\tt open =
  open + 1} is abstracted to the two instructions labeled (C) and (D).
The abstraction is coarse, but still captures the problem. 
Consider two threads {\tt T1} and {\tt T2} running the {\tt
  open\_dev\_abs()} procedure. 
The following trace is possible under a preemptive scheduler, 
but not under a non-preemptive scheduler:
{\tt T1.A; T2.A; T1.B; T1.C; T1.D; T2.B; T2.C; T2.D}.
Moreover, the trace cannot be transformed by swapping independent events 
into any trace possible under a non-preemptive scheduler. This is because 
instructions {\tt A} and {\tt D} are not independent.
Hence, the abstract trace exhibits the problem of two successive
calls to {\tt power\_up()} when executing with a preemptive
scheduler.  
Our synthesis procedure finds this problem, and 
fixes it by introducing a lock in {\tt open\_dev()} (see \secref{algo}).

\section{Preliminaries and Problem Statement}
\label{sec:defs}


\noindent{\bf Syntax.}
We assume that programs are written in a concurrent while language
$\wlang$.
A concurrent program $\cProg$ in $\wlang$ is a finite collection of
threads $\langle \thread_1, \ldots, \thread_n \rangle$ where each thread
is a statement written in the syntax from \autoref{fig:wlang_syntax}.
All $\wlang$ variables (program variables \texttt{std\_var}, lock variables
\texttt{lock\_var}, and condition variable \texttt{cond\_var}) range over integers and each statement
is labeled with a unique location identifier $\loc$.
The only non-standard syntactic constructs in $\wlang$ relate to the
{\em tags}.
Intuitively, each tag is a communication channel between the
program and an interface to an external system, and the $\Input{\sf tag}$ and
$\Output{\sf tag}{expr}$ statements read from and write to the channel.
We assume that the program and the external system interface can only communicate 
through the channel. 
In practice, we use the tags to model device registers.
In our presentation, we consider only a single external interface. Our
implementation can handle communication with several interfaces.

\begin{figure}

  \vspace{1mm}
  \hbra
  \vspace{-2mm}
  \begin{alltt}
expr ::= std\_var | constant | \(operator\)(expr, expr, \(\ldots\), expr) 
lstmt ::= loc: stmt | lstmt; lstmt
stmt ::= skip | std\_var := expr | std\_var := havoc() 
 | if (expr) lstmt else lstmt | while (expr) lstmt | std\_var := input(tag)
 | output(tag, expr) | lock(lock\_var)  | unlock(lock\_var) 
 | signal(cond\_var)  | await(cond\_var) | reset(cond\_var) | yield
  \end{alltt}
  \vspace{-7ex}
  \hket
  \vspace{-2mm}
  \caption{Syntax of $\wlang$ \label{fig:wlang_syntax}}
\end{figure}

\noindent{\bf Semantics.} 
We begin by defining the semantics of a single thread in $\wlang$, and
then extend the definition to concurrent non-preemptive and preemptive
semantics. Note that in our work, reads and writes are assumed to execute 
atomically and further, we assume a sequentially consistent memory
model.

\noindent{\em Single-thread semantics.}
A program state is given by
$\langle \VariableValuation, \Prog \rangle$ where $\VariableValuation$
is a valuation of all program variables, and $\Prog$ is the statement
that remains to be executed.
Let us fix a thread identifier $\tid$.

The operational semantics of a thread executing in isolation
is given in \autoref{fig:single_thread_semantics}.
A single execution step $\smallstepx{\VariableValuation,
\Prog}{\VariableValuation', \Prog'}{\alpha}$ changes the program state
from $\langle\VariableValuation, \Prog\rangle$ to $\langle\VariableValuation', \Prog'\rangle$ while
optionally outputting an {\em observable symbol} $\alpha$.
The absence of a symbol is denoted using $\epsilon$.
Most rules from \autoref{fig:single_thread_semantics} are
standard---the special rules are the {\sc Havoc}, {\sc
Input}, and {\sc Output} rules.
\begin{compactenum}
\item {\sc Havoc}: Statement $\loc: \Asgn{x}{\Havoc}$ assigns $x$ a non-deterministic
  value (say $k$) and outputs the observable $(\tid, \Havoc, k,
  x)$.
\item {\sc Input}, {\sc Output}: $\loc: \Asgn{x}{\Input{t}}$ and $\loc: \Output{t}{\e}$
  read and write values to the channel $t$, and output $(\tid, \mathsf{input}, k, t)$
  and $(\tid,
  \mathsf{output}, k, t)$, where $k$ is the value read
  or written, respectively.
\end{compactenum}

Intuitively, the observables record the sequence of non-deterministic
guesses, as well as the input/output interaction with the tagged
channels. In the following, $\e$ represents an expression and $\evalExpression{\e}$ evaluates an
expression by replacing all variables $v$ with their values in $\VariableValuation$.
\begin{figure}
  \caption{Single thread semantics of $\wlang$\label{fig:single_thread_semantics}}
  \begin{display}{} 
    \staterule{Assign}
    {\evalExpression{\e} = k}
    {\smallstepi{\loc:\ \Asgn{x}{\e}}{\Skip}{\epsilon}{x:=k}}
    \staterule{Havoc}
    {k \in \natset \qquad \alpha = (\tid, \Havoc, k, x)}
    {\smallstepi{\loc:\ \Asgn{x}{\Havoc}}{\Skip}{\alpha}{x:=k}}
    \\[\GAP]
    \staterule{While1}
    {\evalExpression{\e} = \falses}
    {\smallstepo{\loc:\ \eWhile{\e}{s}}{\Skip}{\epsilon}}
    \staterule{While2}
    {\evalExpression{\e} = \trues}
    {\smallstepo{\loc:\ \eWhile{\e}{s}}{\Seq{s}{\eWhile{\e}{s}}}{\epsilon}}
    \\[\GAP]
    \staterule{If1}
    {\evalExpression{\e} = \trues}
    {\smallstepo{\loc:\ \eIf{\e}{s_1}{s_2}}{s_1}{\epsilon}}
    \staterule{If2}
    {\evalExpression{\e} = \falses}
    {\smallstepo{\loc:\ \eIf{\e}{s_1}{s_2}}{s_2}\epsilon}
    \\[\GAP]
    \staterule{Sequence}
    {\smallstepx{\VariableValuation,s_1}{\VariableValuation',s_1'}{\alpha}}
    {\smallstepx{\VariableValuation,\loc:\ \Seq{s_1}{s_2}}{\VariableValuation', \Seq{s_1'}{s_2}}{\alpha}}
    \staterule{Input}
    {k \in \natset \qquad \alpha = (\tid,\mathsf{input}, k,t)}
    {\smallstepi{\loc:\ \Asgn{x}{\Input{t}}}{\Skip}{\alpha}{x:=k}}
    \\[\GAP]
    \staterule{Skip}
    {}
    {\smallstep{\VariableValuation,\loc:\ \Seq{\Skip}{s_2}}{\VariableValuation,{s_2}}}
    \staterule{Output}
    {\evalExpression{\e} = k \qquad \alpha=(\tid,\mathsf{output}, k,t)}
    {\smallstepo{\loc:\ \Output{t}{\e}}{\Skip}{\alpha}}
  \end{display}
  \vspace{-4ex}
\end{figure}

\noindent{\em Non-preemptive semantics.}
The non-preemptive semantics of $\wlang$ is presented in
\ifappendix
Appendix~\ref{app:semantics}.
\else
the full version \cite{fullversion}.
\fi
The non-preemptive semantics ensures that a single thread from the
program keeps executing as detailed above until one of the following
occurs:
\begin{inparaenum}[(a)]
\item the thread finishes execution, or it encounters
\item a {\sf yield} statement, or
\item a {\sf lock} statement and the lock is taken, or
\item an {\sf await} statement and the condition variable is not set.
\end{inparaenum}
In these cases, a context-switch is possible.

\noindent{\em Preemptive semantics.}
The preemptive semantics of a program is obtained from the
non-preemptive semantics by relaxing the condition on context-switches,
and allowing context-switches at all program points (see
\ifappendix
Appendix~\ref{app:semantics}).
\else
full version \cite{fullversion}).
\fi

\subsection{Problem statement}

A {\em non-preemptive observation sequence} of a program
$\cProg$ is a sequence
$\alpha_0\ldots\alpha_k$ if there exist program states
$\ProgState_0^{pre}$, $\ProgState_0^{post}$, \ldots,
$\ProgState_k^{pre}$, $\ProgState_k^{post}$ such that according to the
non-preemptive semantics of $\wlang$, we have:
\begin{inparaenum}[(a)]
\item for each $0 \leq i \leq k$,
  $\smallstepx{\ProgState_i^{pre}}{\ProgState_i^{post}}{\alpha_i}$,
\item for each $0 \leq i < k$, 
  $\smallstepxs{\ProgState_i^{post}}{\ProgState_{i+1}^{pre}}{\epsilon}$,
  and 
\item for the initial state $\ProgState_\iota$ and a final state (i.e., where
  all threads have finished execution) $\ProgState_f$, 
  $\smallstepxs{\ProgState_{\iota}}{\ProgState_0^{pre}}{\epsilon}$ and 
  $\smallstepxs{\ProgState_k^{post}}{\ProgState_{f}}{\epsilon}$.
\end{inparaenum}
Similarly, a {\em preemptive observation sequence} of a program 
$\cProg$ is a sequence
$\alpha_0\ldots\alpha_k$ as above, with the non-preemptive semantics
replaced with preemptive semantics.
We denote the sets of non-preemptive and preemptive observation
sequences of a program $\cProg$ by $\sem{\cProg}^{NP}$ and
$\sem{\cProg}^P$, respectively.

We say that observation sequences $\alpha_0\ldots\alpha_k$
and $\beta_0\ldots\beta_k$ are {\em equivalent} if:
\begin{compactitem}
\item The subsequences of $\alpha_0\ldots\alpha_k$ and
  $\beta_0\ldots\beta_k$ containing only symbols of the form
  $(\tid, \mathsf{Input}, k, t)$ and $(\tid, \mathsf{Output}, k,
  t)$ are equal, and
\item For each thread identifier $\tid$, the subsequences of
  $\alpha_0\ldots\alpha_k$ and $\beta_0\ldots\beta_k$ containing only
  symbols of the form $(\tid, \mathsf{Havoc}, k, x)$ are equal.
\end{compactitem}
Intuitively, observable sequences are equivalent if they have the same
interaction with the interface, and the same non-deterministic choices
in each thread.
For sets of observable sequences $\mathcal{O}_1$ and $\mathcal{O}_2$, we
write $\mathcal{O}_1 \subseteq \mathcal{O}_2$ to denote that each sequence
in $\mathcal{O}_1$ has an equivalent sequence in
$\mathcal{O}_2$.
Given a concurrent program $\cProg$ 
and a synthesized program $\cProg'$ obtained by adding synchronization 
to $\cProg$, the program $\cProg'$ is {\em preemption-safe} w.r.t. $\cProg$ if 
$\sem{\cProg'}^{P} \subseteq \sem{\cProg}^{NP}$.

We are now ready to state our synthesis
problem.
Given a concurrent program $\cProg$, the aim is 
to synthesize a program $\cProg'$, by adding synchronization to $\cProg$,
such that $\cProg'$ is preemption-safe w.r.t. $\cProg$.

\subsection{Language Inclusion Modulo an Independence Relation}
\label{subsec:langincdef}

We reduce the problem of checking if a synthesized solution
is preemption-safe w.r.t. the original program to an
automata-theoretic problem.

\noindent{\bf Abstract semantics for $\wlang$.}
We first define a single-thread abstract semantics for $\wlang$ (\autoref{fig:abstract_semantics}), which 
tracks types of accesses (read or write) to each memory location
while abstracting away their values.
Inputs/outputs to an external interface are modeled as writes to a special memory location (\texttt{dev}).
Even inputs are modeled as writes because in our applications we cannot assume that reads from 
the external interface are free of side-effects. Havocs become
ordinary writes to the variable they are assigned to.
Every branch is taken non-deterministically and tracked. The only constructs
preserved are the lock and condition variables.
The abstract program state consists of the valuations of the lock and
condition variables and the statement that remains to be executed.
In the abstraction, an observable is of the form $(\tid, \{\mathsf{read,write,exit,loop,then,else}\}, v, \loc)$
and observes the type of access (read/write) to variable $v$ and records
non-deterministic branching choices (exit/loop/then/else). The latter are not associated with any variable.

In \autoref{fig:abstract_semantics}, given expression $\e$, the function $\reads{\tid, \e, \loc}$ represents the sequence
$(\tid, \mathsf{read}, v_1, \loc)\cdot\ldots\cdot (\tid, \mathsf{read}, v_n, \loc)$ where $v_1,\ldots, v_n$ are
the variables in $\e$, in the order they are read to evaluate $\e$.

\begin{figure}
  \caption{Single thread abstract semantics of $\wlang$\label{fig:abstract_semantics}}
  \begin{display}{} 
    \staterule{Assign}
    {\alpha=\reads{\tid,\e,\loc}\cdot(\tid, \mathsf{write}, x, \loc)}
    {\smallstepo{\loc:\ \Asgn{x}{\e}}{\Skip}{\alpha}}
    \staterule{Havoc}
    {\alpha=(\tid,\mathsf{write},x,\loc)}
    {\smallstepo{\loc:\ \Asgn{x}{\Havoc}}{\Skip}{\alpha}}
    \\[\GAP]
    \staterule{While1}
    {\alpha=\reads{\tid,\e,\l}\cdot (\tid, \mathsf{exit},\textvisiblespace, \loc)}
    {\smallstepo{\loc:\ \eWhile{\e}{s}}{\Skip}{\alpha}}
    \staterule{While2}
    {\alpha=\reads{\tid,\e,\l}\cdot (\tid, \mathsf{loop},\textvisiblespace , \loc)}
    {\smallstepo{\loc:\ \eWhile{\e}{s}}{\Seq{s}{\eWhile{\e}{s}}}{\alpha}}
    \\[\GAP]
    \staterule{If1}
    {\alpha=\reads{\tid,\e,\l}\cdot (\tid, \mathsf{then},\textvisiblespace, \loc)}
    {\smallstepo{\loc:\ \eIf{\e}{s_1}{s_2}}{s_1}{\alpha}}
    \staterule{If2}
    {\alpha=\reads{\tid,\e,\l}\cdot (\tid, \mathsf{else},\textvisiblespace, \loc)}
    {\smallstepo{\loc:\ \eIf{\e}{s_1}{s_2}}{s_2}\alpha}
    \\[\GAP]
    \staterule{Sequence}
    {\smallstepx{\VariableValuation,s_1}{\VariableValuation',s_1'}{\alpha}}
    {\smallstepx{\VariableValuation,\loc:\ \Seq{s_1}{s_2}}{\VariableValuation', \Seq{s_1'}{s_2}}{\alpha}}
    \staterule{Input}
    {\alpha=(\tid,\mathsf{write},\mathtt{dev},\loc)\cdot(\tid,\mathsf{write},x,\loc)}
    {\smallstepx{\loc:\ \Asgn{x}{\Input{t}}}{\Skip}{\alpha}}
    \\[\GAP]
    \staterule{Skip}
    {}
    {\smallstep{\VariableValuation,\loc:\ \Seq{\Skip}{s_2}}{\VariableValuation,{s_2}}}
    \staterule{Output}
    {\alpha=\reads{\tid,\e,\loc}\cdot(\tid,\mathsf{write},\mathtt{dev},\loc)}
    {\smallstepo{\loc:\ \Output{t}{\e}}{\Skip}{\alpha}}
  \end{display}
  \vspace{-4ex}
\end{figure}

The abstract program semantics (Figures \ref{fig:nonpreemptive_semantics} and \ref{fig:preemptive_semantics}) is the same as the concrete program semantics
where the single thread semantics is replaced by the abstract single thread semantics. Locks and conditionals and operations on them are not abstracted.

As with the concrete semantics of $\wlang$, we can define the
non-preemptive and preemptive observable sequences for abstract
semantics.
For a concurrent program $\cProg$, we denote the sets of abstract
preemptive and non-preemptive observable sequences by
$\sem{\cProg}^{P}_{abs}$ and $\sem{\cProg}^{NP}_{abs}$, respectively.

Abstract observation sequences $\alpha_0\ldots\alpha_k$
and $\beta_0\ldots\beta_k$ are {\em equivalent} if:
\begin{compactitem}
\item For each thread $\tid$, the subsequences of $\alpha_0\ldots\alpha_k$ and $\beta_0\ldots\beta_k$ containing only
  symbols of the form $(\tid,a,v,\loc)$, with $a\in \{\mathsf{read,write,exit,loop,then,else}\}$ are equal,
\item For each variable $v$, the subsequences of
  $\alpha_0\ldots\alpha_k$ and $\beta_0\ldots\beta_k$ containing only
  write symbols (of the form $(\tid, \mathsf{write}, v, \loc)$) are equal, and
\item For each variable $v$, the multisets of symbols of the form 
  $(\tid, \mathsf{read}, v, \loc)$ between any two write symbols, as well
  as before the first write symbol and after the last write symbol are 
  identical.
\end{compactitem}

\noindent We first show that the abstract semantics is sound w.r.t. preemption-safety 
\ifappendix
(see Appendix~\ref{app:proof} for the proof).
\else
(see full version for the proof \cite{fullversion}).
\fi
\begin{restatable}{theorem}{correctness}
\label{thm:correctness}
Given concurrent program $\cProg$ and a synthesized program $\cProg'$ 
obtained by adding synchronization to $\cProg$, 
$\sem{\cProg'}^{P}_{abs} \subseteq \sem{\cProg}^{NP}_{abs}\Rightarrow
\sem{\cProg'}^{P} \subseteq \sem{\cProg}^{NP}$.
\end{restatable}
  
\noindent{\bf Abstract semantics to automata.}
An \nfa\ $\Aut$ is a tuple $(Q, \Sigma, \Delta, Q_\iota, F)$ where
$\Sigma$ is a finite alphabet, $Q,Q_\iota,F$ are finite sets of states, initial states and final states, respectively and $\Delta
$ is a set of transitions.
A word $\sigma_0\ldots\sigma_k \in \Sigma^*$ is {\em accepted} by
$\Aut$ if there exists a sequence of states $q_0\ldots q_{k+1}$
such that $q_0\in Q_\iota$ and $q_{k+1}\in F$ and $\forall i:(q_i, \sigma_i, q_{i+1}) \in \Delta$.
The set of all words accepted by $\Aut$ is called the
language of $\Aut$ and is denoted $\mathcal{L}(\Aut)$.

Given a program $\cProg$, we can construct automata
$\Aut(\sem{\cProg}^{NP}_{abs})$ and $\Aut(\sem{\cProg}^{P}_{abs})$ 
that accept exactly the observable sequences under the respective semantics.
We describe their construction informally.
Each automaton state is a program state of the abstract
semantics
and the alphabet is the set of abstract observable symbols.
There is a transition from one state to another on an observable symbol
(or an $\epsilon$)
iff the program can execute one step under the corresponding
semantics to reach the other state while outputting the observable
symbol (on an $\epsilon$).

\noindent{\bf Language inclusion \upto~an independence relation.}
Let $I$ be a non-reflexive, symmetric
binary relation over an alphabet $\Sigma$. 
We refer to $I$ as the {\em independence relation} and to elements of
$I$ as {\em independent} symbol pairs.
We define a symmetric binary relation $\approx$ over words in
$\Sigma^*$: for all words $\sigma, \sigma'  \in \Sigma^*$ and $(\alpha,
\beta) \in I$, $(\sigma \cdot \alpha \beta \cdot \sigma', \sigma \cdot
\beta \alpha \cdot \sigma') \in \, \approx$.
%
Let $\approx^t$ denote the reflexive transitive closure of
$\approx$.\footnote{The equivalence classes of $\approx^t$ are
Mazurkiewicz traces.}
Given a language ${\cal L}$ over $\Sigma$, the closure of ${\cal L}$
w.r.t. $I$, denoted $\CI({\cal L})$, is the set $\{\sigma \in \Sigma^* {:}\ \exists
\sigma' \in \cal L \text{ with } (\sigma,\sigma') \in \, \mbox{$\approx^t$}\}$. 
Thus, $\CI({\cal L})$ consists of all words that can be obtained from
some word in ${\cal L}$ by repeatedly commuting adjacent independent
symbol pairs from $I$.


\begin{definition}[Language inclusion \upto~an independence relation]
\label{def:langinc}
Given \nfas\ $A,B$ over a common
alphabet $\Sigma$ and an independence relation $I$ over $\Sigma$,
the language inclusion problem \upto~$I$ is: $\lang A
\subseteq \CI(\lang B)$?
\end{definition}

We reduce preemption-safety under the abstract
semantics to 
language inclusion modulo an independence relation.
The independence relation $I$ we use is defined on the set of abstract
observable symbols as follows: $((\tid, a,v,
\loc), (\tid', a',v', \loc')) \in I$ iff 
$\tid \neq \tid'$, and one of the following holds: 
\begin{inparaenum}[(a)]
\item $v \neq v'$ or
\item $a \neq \mathsf{write} \wedge a'\neq
  \mathsf{write}$.
\end{inparaenum}

\begin{proposition}
\label{prop:inclusion}
  Given concurrent programs $\cProg$ and $\cProg'$, 
  $\sem{\cProg'}^{P}_{abs} \subseteq \sem{\cProg}^{NP}_{abs}$ iff
  $\Lang{\Aut(\sem{\cProg'}^{P}_{abs})} \subseteq \CI(\Lang{\Aut(\sem{\cProg}^{NP}_{abs})})$.
\end{proposition}

\section{Checking Language Inclusion}
\label{sec:langinc}
We first focus on the problem of language inclusion \upto~an independence relation (\defref{langinc}).  This question corresponds to preemption-safety (\thmref{correctness}, \propref{inclusion}) and its solution drives our synchronization synthesis (\secref{algo}).

\begin{theorem}\label{thm:langinc.undecidable}
For \nfas~$A, B$ over alphabet $\Sigma$ and an independence relation $I\subseteq \Sigma\times\Sigma$, $\Lang \A \subseteq \CI(\Lang \B)$ is undecidable\cite{bertoni1982equivalence}.  
\end{theorem} 

Fortunately, a bounded version of the problem is decidable. 
Recall the relation $\approx$ over $\Sigma^*$ from \subsecref{langincdef}. 
We define a symmetric binary relation $\approx_i$ over $\Sigma^*$:
$(\sigma, \sigma') \in \, \approx_i$ iff $\exists (\alpha,\beta) \in I$: 
$(\sigma, \sigma') \in \, \approx$, $\sigma[i] = \sigma'[i+1] = \alpha$ and 
$\sigma[i+1] = \sigma'[i] = \beta$.
Thus $\approx^i$ consists of all words that can be optained from each other by commuting 
the symbols at positions $i$ and $i+1$.
We next define a symmetric binary relation $\asymp$ over $\Sigma^*$: 
$(\sigma,\sigma') \in \, \asymp$ iff $\exists \sigma_1,\ldots,\sigma_t$: 
$(\sigma,\sigma_1) \in \, \approx_{i_1},\ldots, (\sigma_{t},\sigma') \in \, \approx_{i_{t+1}}$ 
and $i_1 < \ldots < i_{t+1}$.
The relation $\asymp$ intuitively consists of words obtained from each other
by making a single forward pass commuting multiple pairs of adjacent symbols.
Let $\asymp^k$ denote the $k$-composition of $\asymp$ with itself. 
Given a language ${\cal L}$ over $\Sigma$, 
we use $\CIk({\cal L})$ to denote the set $\{\sigma \in \Sigma^*: \exists \sigma' \in \cal L \text{ with } (\sigma,\sigma') \in \, \asymp^{\mbox{$\scriptstyle k$}}  \}$.
In other words, $\CIk({\cal L})$ consists of all words which can be generated from 
${\cal L}$ using a finite-state transducer that remembers at most $k$ symbols 
of its input words in its states.  

\begin{definition}[Bounded language inclusion \upto~an independence relation]
Given \nfas~$\A, \B$ over $\Sigma$, $I\subseteq \Sigma\times\Sigma$ and a constant $k>0$, the $k$-bounded language inclusion problem \upto~$I$ is: $\lang \A \subseteq \CIk(\lang \B)$? 
\end{definition}

\begin{theorem}\label{thm:boundlanginc}
For \nfas~$\A, \B$ over $\Sigma$, $I\subseteq \Sigma\times\Sigma$ and a constant $k>0$,  
$\lang A \subseteq \CIk(\lang B)$ is decidable. 
\end{theorem}

We present an algorithm to check $k$-bounded language inclusion
\upto~$I$, based on the antichain algorithm for standard language
inclusion \cite{de2006antichains}. 

\noindent{\bf Antichain algorithm for language inclusion}. 
Given a partial order $(X, \sqsubseteq)$, an antichain over $X$ is
a set of elements of $X$ that are incomparable w.r.t. $\sqsubseteq$. In order
to check $\Lang \A \subseteq \CI(\Lang \B)$ for \nfas~$A =
(Q_\A,\Sigma,\Delta_\A,Q_{\iota,\A},F_\A)$ and $B = (Q_\B,\Sigma,\Delta_\B,Q_{\iota,\B},F_\B)$,
the antichain algorithm proceeds by exploring $\A$ and $\B$ in lockstep.
While $\A$ is explored nondeterministically, $\B$ is determinized on the fly 
for exploration. The algorithm maintains
an antichain, consisting of tuples of the form $(s_\A, S_\B)$, where $s_\A \in Q_\A$ and $S_\B \subseteq Q_\B$. 
The ordering relation $\sqsubseteq$ is given by 
$(s_\A, S_\B) \sqsubseteq (s'_\A, S'_\B)$ iff $s_\A = s'_\A$ and $S_\B
\subseteq S'_\B$. The algorithm also maintains a {\em frontier} set of 
tuples {\em yet} to be explored. 

Given state $s_\A \in Q_\A$ and a symbol $\alpha \in \Sigma$,
let $\suc_\alpha(s_\A)$ denote $\{s_\A'
\in Q_\A : (s_\A,\alpha,s_\A') \in \Delta_\A\}$.  Given set of states $S_\B
\subseteq Q_\B$, let $\suc_\alpha(S_\B)$ denote $\{s_\B'\in Q_\B:
\exists s_\B \in S_\B:\ (s_\B,\alpha,s_\B')\in\Delta_\B\}$.  
Given tuple $(s_\A, S_\B)$ in the frontier set, let 
$\suc_\alpha(s_\A, S_\B)$ denote $\{(s'_\A,S'_\B): s'_\A \in
\suc_\alpha(s_\A), S'_\B = \suc_\alpha(s_\B)\}$. 

In each step, the antichain algorithm explores $\A$ and $\B$ by computing
$\alpha$-successors of all tuples in its current frontier set for all possible
symbols $\alpha \in \Sigma$. Whenever a tuple $(s_\A, S_\B)$ is found 
with $s_\A \in F_\A$ and $S_\B \cap F_\B=\emptyset$,
the algorithm reports a counterexample to language inclusion. Otherwise, the
algorithm updates its frontier set and antichain to include 
the newly computed successors using the two rules enumerated below. 
Given a newly computed successor tuple $p'$:
\begin{compactitem}
\item Rule 1: if there exists a tuple $p$ in the antichain 
with $p \sqsubseteq p'$, then $p'$ is not added to the frontier set or antichain,  
\item Rule 2: else, if there exist tuples $p_1, \ldots, p_n$ in the  antichain 
with $p' \sqsubseteq p_1, \ldots, p_n$, then $p_1, \ldots, p_n$ are removed from the antichain.
\end{compactitem}
The algorithm terminates by either reporting a counterexample, or by declaring
success when the frontier becomes empty.  

\noindent{\bf Antichain algorithm for $k$-bounded language inclusion \upto~$I$}. 
This algorithm is essentially the same as the standard antichain algorithm, with 
the automaton $\B$ above replaced by an automaton 
$\autCB$ accepting $\CIk(\lang \B)$. 
The set $Q_{\autCB}$ of states of $\autCB$ consists of triples $(s_\B, \early, \late)$, 
where $s_\B \in Q_\B$ and $\early, \late$ are $k$-length words over $\Sigma$. 
Intuitively, the words $\early$ and $\late$ store symbols that are expected to be matched later along a run. 
The set of initial states of $\autCB$ is $\{(s_\B,\emptyset,\emptyset): s_\B \in I_\B\}$. 
The set of final states of $\autCB$ is  $\{(s_\B,\emptyset,\emptyset): s_\B \in F_\B\}$. 
The transition relation $\Delta_{\autCB}$ is constructed by repeatedly applying the following rules, in order,  
for each state $(s_\B, \early, \late)$ and each symbol $\alpha$. In what follows, $\eta[\setminus i]$ 
denotes the word obtained from $\eta$ by removing its $i^{th}$ symbol. 
\begin{compactenum}
\item Pick {\em new} $s'_\B$ and $\beta\in\Sigma$ such that $(s_\B, \beta, s_\B') \in \Delta_\B$
\item 
\begin{inparaenum}[(a)]
	\item If $\forall i$: $\early[i] \neq \alpha$ and 
        $\alpha$ is independent of all symbols in $\early$, \\
        $\late' \, \assign \, \late \cdot \alpha$ and $\early' \, \assign \, \early$,
        \item else, if $\exists i$:  $\early[i] = \alpha$ and 
	$\alpha$ is independent of all symbols in $\early$ prior to $i$, 
        $\early' \, \assign \, \early[\setminus i]$ and $\late'\, \assign\, \late$
        \item else, go to 1
   \end{inparaenum}
\item 
\begin{inparaenum}[(a)]
        \item If $\forall i$: $\late'[i] \neq \beta$ and
	$\beta$ is independent of all symbols in $\late'$,\\
        $\early' \, \assign \early' \,\cdot \beta$, 
        \item else, if $\exists i$: $\late'[i] = \beta$ and 
	$\beta$ is independent of all symbols in $\late'$ prior to $i$, 
	$\late' \, \assign \, \late'[\setminus i]$
	\item else, go to 1
\end{inparaenum}
\item Add $((s_\B, \early, \late),\alpha,(s'_\B, \early', \late'))$ to $\Delta_{\autCB}$ and go to 1.
\end{compactenum}

\begin{example} 
In \figref{langincex}, we have an \nfa~$\B$ with $\lang \B =  \{\alpha\beta, \beta\}$, 
$I = \{(\alpha,\beta)\}$ and $k = 1$. The states of $\autCB$ are triples $(q, \early, \late)$, where 
$q \in Q_\B$ and $\early, \late \in \{\emptyset,\alpha,\beta\}$. 
We explain the derivation of a couple of transitions of $\autCB$. The transition shown in bold 
from $(q_0, \emptyset,\emptyset)$ on symbol $\beta$ is obtained by applying the following rules once: 
1. Pick $q_1$ since $(q_0, \alpha, q_1) \in \Delta_\B$. 2(a). $\late'\ \assign\ \beta$, $\early'\ \assign\ \emptyset$. 
3(a). $\early'\ \assign\ \alpha$. 4. Add $((q_0, \emptyset, \emptyset),\beta,(q_1, \alpha, \beta))$ to $\Delta_{\autCB}$. 
The transition shown in bold from $(q_1, \alpha,\beta)$ on symbol $\alpha$ 
is obtained as follows:
1. Pick $q_2$ since $(q_1, \beta, q_2) \in \Delta_\B$. 2(b). $\early'\ \assign\ \emptyset$, $\late'\ \assign\ \beta$.  
3(b). $\late'\ \assign\ \emptyset$. 4. Add $((q_1, \alpha, \beta),\beta,(q_2, \emptyset, \emptyset))$ to $\Delta_{\autCB}$.
It can be seen that $\autCB$ accepts the language $\{\alpha\beta,\beta\alpha,\beta\} = \CIk(\B)$.
\end{example}

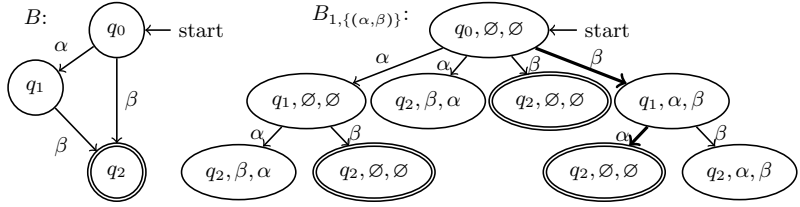
\begin{figure}
\begin{center}
\footnotesize{
\begin{tikzpicture}[->,node distance=1.25cm,semithick,scale=0.9, every node/.style={transform shape}]
\tikzstyle{every state}=[ellipse,text=black]

 \node[state] (q0)              {$q_0$};
 \node[right of=q0] (start) {start};
 \node[state,below of=q0,xshift=-1.2cm,yshift=4mm] (q1)   {$q_1$};
 \node[state,accepting,below of=q0, node distance=2.5cm,yshift=4mm]  (q2)   {$q_2$};
 
 \node[state,draw=none] (B) [left of=q0,node distance=1.2cm,yshift=0.2cm]  {\small{$\B$:}};

 \path (q0) edge [above left] node {$\alpha$} (q1)
 (start) edge (q0)
         (q0)   edge [right] node {$\beta$} (q2)
  (q1) edge node [below left] {$\beta$} (q2);

 \node[state,right of=q0,node distance=5.5cm] (r0)              {$q_0,\emptyset,\emptyset$};
 \node[right of=r0,xshift=0.5cm] (start2) {start};
 \node[state,below of=r0,xshift=-2.7cm,yshift=2mm] (r1)   {$q_1,\emptyset,\emptyset$};
 \node[state,below of=r0,xshift=-0.9cm,yshift=2mm] (r1a)     {$q_2,\beta,\alpha$};
 \node[state,accepting,below of=r0,xshift=0.9cm,yshift=2mm] (r1b)   {$q_2,\emptyset,\emptyset$};
 \node[state,below of=r0,xshift=2.7cm,yshift=2mm] (r1c)   {$q_1,\alpha,\beta$};

 \node[state,below of=r1,xshift=-1cm,yshift=2mm]  (r2)   {$q_2,\beta,\alpha$};
 \node[state,accepting,below of=r1,xshift=1cm,yshift=2mm]  (r2a)   {$q_2,\emptyset,\emptyset$};
 \node[state,accepting,below of=r1c,xshift=-1cm,yshift=2mm]  (r2b)   {$q_2,\emptyset,\emptyset$};
 \node[state,below of=r1c,xshift=1cm,yshift=2mm]  (r2c)   {$q_2,\alpha,\beta$};

 \node[state,draw=none] (autCB) [left of=r0,node distance=1.9cm,yshift=0.2cm]  {\small{$\B_{1,\{(\alpha,\beta)\}}$:}};
 
 \path (r0) edge [left] node[yshift=1mm] {$\alpha$} (r1)
 (start2) edge (r0)
           (r0) edge [left] node {$\alpha$} (r1a)
            edge [right] node {$\beta$} (r1b)
            edge [right,very thick] node[yshift=1mm] {$\beta$} (r1c)
       (r1) edge node [left,yshift=0.3mm] {$\alpha$} (r2)
            edge node [right,yshift=0.3mm] {$\beta$} (r2a)
       (r1c) edge [very thick,yshift=0.3mm] node [left] {$\alpha$} (r2b)
             edge node [right,yshift=0.3mm] {$\beta$} (r2c);

\end{tikzpicture}
}
\end{center}
\vspace{-5mm}
\caption{Example for illustrating construction of $\autCB$ for $k = 1$ and $I = \{(\alpha, \beta)\}$.}
\label{fig:langincex}
\end{figure}

\begin{proposition}
Given $k>0$, \nfa~$\autCB$ described above accepts $\CIk(\lang \B)$.
\end{proposition}

We develop a procedure to check 
language inclusion \upto~$I$ by iteratively increasing the bound $k$
\ifappendix
(see Appendix~\ref{app:algo}
\else
(see the full version~\cite{fullversion}
\fi
for the complete algorithm). 
The procedure is {\em incremental}: the check for $k+1$-bounded language
inclusion \upto~$I$ only explores paths along which the bound $k$
was exceeded in the previous iteration.

\section{Synchronization Synthesis}
\label{sec:algo}

We now present our iterative synchronization synthesis procedure,
which is based on the procedure in~\cite{POPL15}.  The reader is
referred to~\cite{POPL15} for further details. 
The synthesis procedure starts with the original program 
$\cProg$ and in each iteration generates a candidate synthesized 
program $\cProg'$. The candidate $\cProg'$ is
checked for preemption-safety w.r.t. $\cProg$  
under the abstract semantics,  
using our procedure for bounded language inclusion modulo $I$. 
If $\cProg'$ is found preemption-safe w.r.t. $\cProg$ under the abstract semantics, 
the synthesis procedure outputs $\cProg'$. 
Otherwise, an abstract counterexample $\cex$ is obtained. The counterexample is  
analyzed to infer additional synchronization to be added to $\cProg'$ 
for generating a new synthesized candidate.

\comment{
\begin{algorithm}
  \caption{The synthesis procedure\label{algo:main_algo}}
  \begin{algorithmic}[1]
    \Require Concurrent program $\cProg$
    \Ensure Concurrent program $\cProg'$ such that
    $\sem{\cProg'}^{P}_{abs} \equiv \sem{\cProg'}^{NP}_{abs}$
    \State $\cProg' \gets \cProg$
    \While {$\true$}
    \State $\Aut_{P} = \Aut(\sem{\cProg'}^P_{abs})$
    \State $\Aut_{NP} = \Aut(\sem{\cProg}^{NP}_{abs})$
    \If {$\mathit{LanguageInclusionUptoIndependence}(\Aut_{P}, \Aut_{NP})$} \Comment {\algoref{inclusionnew}}
    \State \Return $\cProg'$
    \Else
    \State $\cex \gets \mathit{CexToInclusion}(\Aut_{P}, \Aut_{NP})$
    \State $\Phi \gets \mathit{PreemptionFreeVariations}(\cex)$
    \State $\rho \gets \mathit{Generalize}(\Phi)$
    \State $\cProg' \gets \mathit{InferFixes}(\cProg', \rho)$
    \EndIf
    \EndWhile
  \end{algorithmic}
\end{algorithm}
}


The counterexample trace $\cex$ is a sequence of event identifiers: 
$\tid_0.\loc_0 ; \ldots ; \tid_n.\loc_n$, 
where each $\loc_i$ is a location identifier.
We first analyze the {\em neighborhood} of $\cex$, denoted 
$\nhood(\cex)$, consisting of traces that are permutations of the events in $\cex$. 
Note that each trace corresponds to an abstract observation sequence. 
Furthermore, note that preemption-safety requires the abstract
observation sequence of any trace in  
$\nhood(\cex)$ to be equivalent to that of some trace in $\nhood(\cex)$ feasible
under non-preemptive semantics. Let {\em bad traces} refer to the
traces in $\nhood(\cex)$  
that are feasible under preemptive semantics and do not meet the 
preemption-safety requirement. The goal of our counterexample analysis 
is to characterize all bad traces in $\nhood(\cex)$ in order to enable
inference of synchronization fixes.   

In order to succinctly represent subsets of $\nhood(\cex)$, we use {\em
ordering constraints}.
Intuitively, ordering constraints are of the following forms:
\begin{inparaenum}[(a)]
\item atomic constraints $\Phi = A < B$ where $A$ and $B$ are events
  from $\cex$. The constraint $A < B$ represents the set of traces in
  $\nhood(\cex)$ where event $A$ is scheduled before event $B$;
\item Boolean combinations of atomic constraints $\Phi_1 \wedge \Phi_2$,
  $\Phi_1 \vee \Phi_2$ and $\neg \Phi_1$. We have that $\Phi_1 \wedge
  \Phi_2$ and $\Phi_1 \vee \Phi_2$ respectively represent the
  intersection and union of the set of  traces represented by $\Phi_1$
  and $\Phi_2$, and that $\neg \Phi_1$ represents the complement
  (with respect to $\nhood(\cex)$) of the traces represented by
  $\Phi_1$.
\end{inparaenum}

\noindent{\bf Non-preemptive neighborhood}.
First, we generate all traces in $\nhood(\cex)$ that are feasible under
non-preemptive semantics.
We represent a single trace $\pi$ using an ordering constraint $\Phi_\pi$ that 
captures the ordering between non-independent accesses to
variables in $\pi$. 
We represent all traces in $\nhood(\cex)$ that are feasible under
non-preemptive semantics using the expression $\Phi = \bigvee_{\pi}
\Phi_\pi$. The expression $\Phi$ acts as the correctness specification
for traces in $\nhood(\cex)$.

\noindent {\em Example}. Recall the counterexample trace from the running example in \secref{runningexample}:  $\cex= \mathtt{T1.A; T2.A; T1.B; T1.C; T1.D; T2.B; T2.C; T2.D}$. 
There are two trace in $\nhood(\cex)$ that are feasible under non-preemptive semantics: $\pi_1=\mathtt{T1.A;T1.B;T1.C;T1.D;T2.A;T2.B;T2.C;T2.D}$
and $\pi_2=\mathtt{T2.A;T2.B;T2.C;T2.D;T1.A;T1.B;T1.C;T1.D}$. We represent $\pi_1$ 
as $\Phi(\pi_1)= \{\mathtt{T1.A,T1.C,T1.D}\} <\mathtt{T2.D} \; \wedge \; 
\mathtt{T1.D}<\{\mathtt{T2.A,T2.C,T2.D}\} \; \wedge \; \mathtt{T1.B}<\mathtt{T2.B}$
and $\pi_2$ as $\Phi(\pi_2) = \mathtt{T2.D} < \{\mathtt{T1.A,T1.C,T1.D}\} \; \wedge \; 
\{\mathtt{T2.A,T2.C,T2.D}\} < \mathtt{T1.D} \; \wedge \; \mathtt{T2.B}<\mathtt{T1.B}$.
The correctness specification is $\Phi = \Phi(\pi_1) \vee \Phi(\pi_2)$.

\noindent{\bf Counterexample generalization}.
We next build a quantifier-free first order formula $\Psi$ over the event identifiers in $\cex$ 
such that any model of $\Psi$ corresponds to a bad trace in $\nhood(\cex)$.
We iteratively enumerate models $\pi$ of $\Psi$, building 
a constraint $\rho= \Phi(\pi)$ for each model $\pi$, and 
generalizing each $\rho$ into $\rho_g$ to represent a larger set of
bad traces. 

\noindent {\em Example}. Our trace $cex$ from  \secref{runningexample} 
would be generalized to $\mathtt{T2.A}<\mathtt{T1.D} \wedge \mathtt{T1.D}<\mathtt{T2.D}$.
Any trace that fulfills this constraint is bad.

\noindent{\bf Inferring fixes}.
From each generalized formula $\rho_g$ described above, 
we infer possible synchronization fixes to eliminate all bad traces satisfying $\rho_g$. 
The key observation we exploit is that common concurrency bugs often show up in 
our formulas as simple patterns of ordering constraints between events. 
For example, the pattern $\tid_1.\loc_1 < \tid_2.\loc_2  \; \wedge \; 
\tid_2.\loc'_2 < \tid_1.\loc'_1$ indicates an atomicity violation and 
can be rewritten into $\mtt{lock}(\tid_1.[\loc_1:\loc'_1],\tid_2.[\loc_2:\loc'_2])$.
The complete list of such rewrite rules is 
\ifappendix
in Appendix~\ref{app:patterns}.
\else
in the full version~\cite{fullversion}.
\fi
This list includes  
inference of locks and reordering of notify statements.  
The set of patterns we use for synchronization inference are not
complete, i.e., there might be generalized formulae $\rho_g$ that are
not matched by any pattern.  In practice, we found our current set of
patterns to be adequate for most common concurrency bugs, including all
bugs from the benchmarks in this paper. Our technique and tool
can be easily extended with new patterns.

\noindent{\em Example}. The generalized constraint 
$\mathtt{T2.A < T1.D \; \wedge \; T1.D<T2.D}$ matches the lock rule and
yields $\mtt{lock(T2.[A:D],T1.[D:D])}$.
Since the lock involves events in the same function, the lock is merged into a 
single lock around instructions $\mtt{A}$ and $\mtt{D}$ in \texttt{open\_dev\_abs}.
This lock is not sufficient to make the program preemption-safe. 
Another iteration of the synthesis procedure generates another counterexample for analysis 
and synchronization inference.  

\comment{
\noindent{\bf Preserving traces}.
Our synthesis procedure will yield a program $\cProg'$ that is preemption-safe w.r.t.
$\cProg$. To ensure $\cProg'$ and $\cProg$ are also preemption-equivalent, 
we restrict the placement of synthesized locks and reordering of notify statements syntactically 
--- we do not permit the placement of locks  and reordering across preemption points ({\tt yield}, {\tt lock}, {\tt await}).
If such a lock is synthesized, then it is unlocked before the preemption point and locked again afterwards.
}
\begin{proposition}
 \label{prop:preemptequiv}
 If our synthesis procedure generates a program $\cProg'$, then
 $\cProg'$ is preemption-safe with respect to $\cProg$.
\end{proposition}
Note that our procedure does not guarantee that the synthesized program
$\cProg'$ is deadlock-free.
However, we avoid obvious deadlocks using heursitics such as merging
overlapping locks.
Further, our tool supports detection of any additional deadlocks
introduced by synthesis, but relies on the user to fix them.

\section{Implementation and Evaluation}
\label{sec:impl}

We implemented our synthesis procedure in \ourtool.
\ourtool\ is 
comprised of 5000 lines of C++ code and uses Clang/LLVM and Z3 as 
libraries. It is available as open-source software along with benchmarks at \url{https://github.com/thorstent/Liss}.
The language inclusion algorithm is available separately as a library called {\sc Limi} (\url{https://github.com/thorstent/Limi}).
\ourtool\ implements the synthesis method presented in this paper 
with several optimizations.  For example, we take advantage of the fact 
that language inclusion violations can often be detected by 
exploring only a small fraction of the input automata by 
constructing $\Aut(\sem{\cProg}^{NP}_{abs})$ and 
$\Aut(\sem{\cProg}^{P}_{abs})$ on the fly.

Our prototype implementation has several limitations.  First, 
\ourtool\ uses function inlining and therefore cannot handle 
recursive programs.  Second, we do not implement any form of alias 
analysis, which can lead to unsound abstractions.  For example, we 
abstract statements of the form ``\texttt{*x = 0}'' as writes to 
variable \texttt{x}, while in reality other variables can be 
affected due to pointer aliasing.  We sidestep this issue by 
manually massaging input programs to eliminate aliasing.  

Finally, \ourtool\ implements a simplistic lock insertion 
strategy.  Inference rules in Figure~\ref{fig:fix_patterns} 
produce locks expressed as sets of instructions that should be 
inside a lock.  Placing the actual lock and unlock instructions in 
the C code is challenging because the instructions in the trace 
may span several basic blocks or even functions. We follow a 
structural approach where we find the innermost common parent 
block for the first and last instructions of the lock and place 
the lock and unlock instruction there. This does not work if the 
code has gotos or returns that could cause control to jump over 
the unlock statement.  At the moment, we simply report such 
situations to the user.  

We evaluate our synthesis method against the following criteria:
\begin{inparaenum}[(1)]
    \item Effectiveness of synthesis from implicit specifications;
    \item Efficiency of the proposed synthesis procedure;
    \item Precision of the proposed coarse abstraction scheme on 
        real-world programs.
\end{inparaenum}

\noindent{\bf Implicit vs explicit synthesis~~}
In order to evaluate the effectiveness of synthesis from implicit 
specifications, we apply \ourtool\ to the set of benchmarks used in  
our previous {\sc ConRepair} tool for assertion-based synthesis~\cite{CAV14}.
In addition, we evaluate \ourtool\ and {\sc ConRepair} 
on several {\em new} assertion-based 
benchmarks (Table~\ref{table:experiments}).  
The set includes microbenchmarks modeling typical concurrency bug 
patterns in Linux drivers and the \texttt{usb-serial} 
macrobenchmark, which models a complete synchronization skeleton 
of the USB-to-serial adapter driver.  We preprocess these 
benchmarks by eliminating assertions used as explicit 
specifications for synthesis.  In addition, we replace statements 
of the form \texttt{assume(v)} with \texttt{await(v)}, redeclaring 
all variables \texttt{v} used in such statements as condition
variables.  This is necessary as our program syntax does
not include \texttt{assume} statements.

\begin{table}[t]
\footnotesize
\newcommand{\notime}{\textless 1s}
\renewcommand{\thefootnote}{{\it\alph{footnote}}}
\small
\begin{tabular}{l || >{\raggedright}p{0.07\textwidth} | p{0.07\textwidth} | p{0.07\textwidth} | p{0.07\textwidth} | p{0.09\textwidth} | p{0.09\textwidth} | p{0.09\textwidth} || p{0.09\textwidth}}
 Name & LOC & Th & It & MB  & BF(s) & Syn(s) & Ver(s) & CR(s) \\
 \hline
 \multicolumn{9}{c}{ConRepair benchmarks \cite{CAV14}} \\
 \hline
 ex1.c & 18 & 2 & 1 & 1\ & \notime & \notime & \notime & \notime\\
 ex2.c & 23 & 2 & 1 & 1\ & \notime & \notime & \notime & \notime \\
 ex3.c & 37 & 2 & 1 & 1\ & \notime & \notime & \notime & \notime \\
 ex5.c & 42 & 2 & 3 & 1\ & \notime & \notime & 2s & \notime \\
 lc-rc.c & 35 & 4 & 0 & 1\ & - & - & \notime & 9s \\
 dv1394.c & 37 & 2 & 1 & 1\ & \notime & \notime & \notime & 17s \\
 em28xx.c & 20 & 2 & 1 & 1\ & \notime & \notime & \notime & \notime \\
 f\_acm.c & 80 & 3 & 1 & 1\ & \notime & \notime & \notime & 1871.99s \\
 i915\_irq.c & 17 & 2 & 1 & 1\ & \notime & \notime & \notime & 2.6s \\
 ipath.c & 23 & 2 & 1 & 1\ & \notime & \notime & \notime & 12s \\
 iwl3945.c & 26 & 3 & 1 & 1\ & \notime & \notime & \notime & 5s \\
 md.c & 35 & 2 & 1 & 1\ & \notime & \notime & \notime & 1.5s \\
 myri10ge.c & 60 & 4 & 2 & 1\ & - & - & \notime & 1.5s \\
 usb-serial.bug1.c & 357 & 7 & 2 & 1\ & 0.4s & 3.1s & 3.4s & \TO\footnotemark[2] \\
 usb-serial.bug2.c & 355 & 7 & 1 & 3 & 0.7s & 2.1s & 12.9s & 3563s \\
 usb-serial.bug3.c & 352 & 7 & 1 & 4 & 3.8s & 1.3s & 111.1s & \TO\footnotemark[2] \\
 usb-serial.bug4.c & 351 & 7 & 1 & 4 & 93.9s & 2.4s & 123.1s & \TO\footnotemark[2] \\
 usb-serial.c\footnotemark[1]  & 357 & 7 & 1 & 4 & - & - & 103.2s & 1200s \\
 \hline
 \multicolumn{9}{c}{CPMAC driver benchmark} \\
 \hline
 cpmac.bug1.c & 1275 & 5 & 1 & 1\ & 1.3s & 113.4s & 21.9s & - \\
 cpmac.bug2.c & 1275 & 5 & 1 & 1\ & 3.3s & 68.4s & 27.8s & - \\
 cpmac.bug3.c & 1270 & 5 & 1 & 1\ & 5.4s & 111.3s & 8.7s & - \\
 cpmac.bug4.c & 1276 & 5 & 2 & 1\ & 2.4s & 124.8s & 31.5s & -\\
 cpmac.bug5.c & 1275 & 5 & 1 & 1\ & 2.8s & 112.0s & 58.0s & -\\
 cpmac.c\footnotemark[1] & 545\footnotemark[3] & 5 & 1 & 1\ & - & - & 17.4s & -\\
\end{tabular}
\vspace{-2mm}
\begin{flushleft}
Th=Threads, It=Iterations, MB=Max bound, BF=Bug finding, Syn=Synthesis, Ver=Verification, Cr={\sc ConRepair}  \hspace{5mm}
\footnotemark[1] bug-free example \hspace{5mm} \footnotemark[2] timeout after 3 hours \\
\end{flushleft}
\vspace{-5mm}
\caption{Experiments}
\label{table:experiments}
\end{table}

We use \ourtool\ to synthesize a preemption-safe version of
each benchmark.  This method is based on the assumption that the 
benchmark is correct under non-preemptive scheduling and bugs can 
only arise due to preemptive scheduling.  We discovered two 
benchmarks (\texttt{lc-rc.c} and \texttt{myri10ge.c}) that violated this assumption, i.e., 
they contained race conditions
that manifested under non-preemptive scheduling; \ourtool\ 
did not detect these race conditions.  \ourtool\ was able to detect and fix 
all other known races without relying on assertions.  Furthermore, 
\ourtool\ detected a new race in the \texttt{usb-serial} family of 
benchmarks, which was not detected by {\sc ConRepair} due to a missing 
assertion.  We compared the output of \ourtool\ with manually placed 
synchronization (taken from real bug fixes) and found that the two 
versions were similar in most of our examples.  

\noindent{\bf Performance and precision.}
{\sc ConRepair} uses CBMC for verification and counterexample 
generation. Due to the coarse abstraction we use, both steps are 
much cheaper with \ourtool.  For example, verification of {\tt 
usb-serial.c}, which was the most complex in our set of 
benchmarks, took \ourtool\ 82 seconds, whereas it took 
{\sc ConRepair} 20 minutes~\cite{CAV14}.  

The loss of precision due to abstraction may cause the inclusion 
check to return a counterexample that is spurious in the concrete
program, leading to unnecessary 
synchronization being synthesized.  On our existing benchmarks,
this only occurred once in the \texttt{usb-serial} driver, where 
abstracting away the return value of a function led to an 
infeasible trace.  We refined the abstraction manually by 
introducing a condition variable to model the return value.  

While this result is encouraging, synthetic benchmarks are not 
necessarily representative of real-world performance.  We 
therefore implemented another set of benchmarks based on a 
complete Linux driver for the TI AR7 CPMAC Ethernet controller.  
The benchmark was constructed as follows.  We manually 
preprocessed driver source code to eliminate pointer aliasing.  We 
combined the driver with a model of the OS API and the software 
interface of the device written in C.  We modeled most OS API 
functions as writes to a special memory location.
Groups of unrelated functions were modeled using 
separate locations.  Slightly more complex models were required 
for API functions that affect thread synchronization.  For 
example, the \texttt{free\_irq} function, which disables the 
driver's interrupt handler, blocks waiting for any outstanding 
interrupts to finish.  Drivers can rely on this behavior to avoid 
races.  We introduced a condition variable to model this 
synchronization.
Similarly, most device accesses were modeled as writes to a 
special \texttt{ioval} variable.  Thus, the only part of the device that 
required a more accurate model was its interrupt enabling logic, 
which affects the behavior of the driver's interrupt handler 
thread.

Our original model consisted of eight threads.   \ourtool\ ran out
of memory on this model, so we simplified it to five threads by 
eliminating parts of driver functionality.  Nevertheless, we 
believe that the resulting model represents the most complex 
synchronization synthesis case study, based on real-world code, 
reported in the literature.  

The CPMAC driver used in this case study did not contain any known 
concurrency bugs, so we artificially simulated five typical race 
conditions that commonly occur in drivers of this 
type~\cite{CAV13}.  \ourtool\ was able to detect and automatically 
fix each of these defects (bottom part of 
Table~\ref{table:experiments}).  We only encountered two program 
locations where manual abstraction refinement was necessary.

We conclude that (1) our coarse abstraction is highly precise in 
practice; (2) manual effort involved in synchronization synthesis 
can be further reduced via automatic abstraction refinement; (3) 
additional work is required to improve the performance of our 
method to be able to handle real-world systems without 
simplification.  In particular, our analysis indicates that 
significant speed-up can be obtained by incorporating a partial 
order reduction scheme into the language inclusion algorithm.

\section{Conclusion}
We believe our approach and the encouraging experimental results open
several directions for future research. Combining
the abstraction refinement, verification (checking language
inclusion modulo an independence relation), and synthesis (inserting
synchronization) more tightly could bring improvements in efficiency.  
An additional direction we plan on exploring is automated handling of
deadlocks, i.e., extending our technique to automatically synthesize
deadlock-free programs.
Finally, we plan to further develop our prototype tool and apply it to
other domains of concurrent systems code. 

\clearpage

\bibliographystyle{splncs03}
\bibliography{refs}

\clearpage

\begin{appendix}

\section{Semantics of preemptive and non-preemptive execution}
\label{app:semantics}

In \figref{nonpreemptive_semantics} we present the non-preemptive semantics.
The preemptive semantics consist of the rules of the non-preemptive semantics
and the single rule in \figref{preemptive_semantics}.

We denote the state of a program as $\langle \VariableValuation, 
ctid, (\Prog_1, \ldots, \Prog_n) \rangle$ where
\begin{inparaenum}[(a)]
\item Valuation $\VariableValuation$ is a valuation of all program
  variables.
  Further, for each lock $l$, we have that $\VariableValuation[l]$ holds
  the identifier of the thread that currently holds the lock, or $0$ if
  no thread holds the lock.
  Similarly, for a condition variable $c$, we have that
  $\VariableValuation[c] = 0$ if the variable is reset and
  $\VariableValuation[c] = 1$ otherwise.
\item The value $ctid$ is the thread identifier of the current
  executing thread or $0$ in the initial state, and
\item Program fragments $\Prog_1$ to $\Prog_n$ are the parts of the
  program to be executed by $\thread_1$ to $\thread_n$, respectively.
\end{inparaenum}

The premise in rule {\sc Sequential} refers to the single-threaded semantics in \figref{single_thread_semantics} or the abstract
single-threaded semantics in \figref{abstract_semantics}.
Rules {\sc LockYield} and {\sc AwaitYield} force a context-switch iff the lock is not available or the condition variable is not set.

\begin{figure}
  \caption{Operational non-preemptive semantics\label{fig:nonpreemptive_semantics}}
  \begin{display}{}
   \staterule{ScheduleStart}
    { ctid = 0 \quad 1 \leq ctid' \leq n }
    {\smallstepx{\VariableValuation,   ctid, (
    \Prog_1, \ldots, \Prog_n)}{\VariableValuation,  
      ctid', (\Prog_1 ,\ldots, \Prog_n)}{\epsilon}}
    \\[\GAP]
    \staterule{Sequential}
    { ctid = i \qquad \smallstepx{\VariableValuation , \Prog_i}{\VariableValuation',
    \Prog_i'}{\alpha} }
    {\smallstepx{\VariableValuation,   ctid, (
    \Prog_1, \ldots, \Prog_i, \ldots, \Prog_n)}{\VariableValuation,  
      ctid, (\Prog_1 ,\ldots, \Prog_i', \ldots, \Prog_n)}{\alpha}}
    \\[\GAP]
    \staterule{LockYield}
    {ctid=i 
     \qquad \VariableValuation(l)\notin \{0,i\} \qquad 1 \leq ctid' \leq n}
    {\smallstep{\VariableValuation,   ctid, (
      \Prog_1,\ldots, \Lock{l}, \ldots, \Prog_n)}{\VariableValuation,  
        ctid', (\Prog_1, \ldots, \Lock{l}, \ldots, \Prog_n)}}
    \\[\GAP]
    \staterule{Lock}
    { ctid=i \qquad \VariableValuation(l)\in \{0,i\}}
    {\smallstep{\VariableValuation, ctid, (\Prog_1,
      \ldots, \Lock{l}, \ldots, \Prog_n)}{\VariableValuation[l := i],  
        ctid, (\Prog_1, \ldots, \Skip, \ldots, \Prog_n)}}
    \\[\GAP]
    \staterule{Unlock}
    { ctid = i \qquad \VariableValuation(l) = ctid}
    {\smallstep{\VariableValuation,  ctid, (\Prog_1,
      \ldots, \Unlock{l}, \ldots, \Prog_n)}{\VariableValuation[l := 0], 
        ctid, (\Prog_1, \ldots, \Skip, \ldots, \Prog_n)}}
    \\[\GAP]
    \staterule{AwaitYield}
    { ctid=i \qquad \VariableValuation(c)= \falses \qquad 1 \leq ctid' \leq n}
    {\smallstep{\VariableValuation, ctid, (\Prog_1,
      \ldots, \Await{c}, \ldots, \Prog_n)}{\VariableValuation,
    ctid, (\Prog_1, \ldots, \Await{c}, \ldots, \Prog_n)}}
    \\[\GAP]
    \staterule{Await}
    { ctid=i \qquad \VariableValuation(c)= \trues}
    {\smallstep{\VariableValuation, ctid, (\Prog_1,
      \ldots, \Await{c}, \ldots, \Prog_n)}{\VariableValuation,
    ctid, (\Prog_1, \ldots, \Skip, \ldots, \Prog_n)}}
    \\[\GAP]
    \staterule{Signal}
    {ctid=i}
    {\smallstep{\VariableValuation,  ctid, (\Prog_1,
      \ldots, \Signal{c}, \ldots, \Prog_n)}{\VariableValuation[c := \trues],
    ctid, (\Prog_1, \ldots, \Skip, \ldots, \Prog_n)}}
    \\[\GAP]
    \staterule{Reset}
    {ctid=i}
    {\smallstep{\VariableValuation,  ctid, (\Prog_1,
      \ldots, \Reseta{c}, \ldots, \Prog_n)}{\VariableValuation[c := \falses],
    ctid, (\Prog_1, \ldots, \Skip, \ldots, \Prog_n)}}
    \\[\GAP]
    \staterule{Sequence}
    { s_1\neq\Skip \qquad
    \smallstepx{\VariableValuation,   ctid, (
    \Prog_1, \ldots, s_1, \ldots, \Prog_n)}{\VariableValuation,  
      ctid, (\Prog_1 ,\ldots,s_1', \ldots, \Prog_n)}{\alpha}}
    {\smallstepx{\VariableValuation,   ctid, (
    \Prog_1, \ldots, \Seq{s_1}{s_2}, \ldots, \Prog_n)}{\VariableValuation,  
      ctid, (\Prog_1 ,\ldots, \Seq{s_1'}{s_2}, \ldots, \Prog_n)}{\alpha}}
    \\[\GAP]
    \staterule{DescheduleSkip}
    {ctid = i \qquad 1 \leq ctid' \leq n \qquad \Prog_i = \Skip}
    {\smallstep{\VariableValuation,   ctid, (\Prog_1,
      \ldots, \Prog_i, \ldots, \Prog_n)}{\VariableValuation,  
    ctid', (\Prog_1, \ldots, \Prog_i, \ldots, \Prog_n)}}
    \\[\GAP]
    \staterule{Yield}
    { ctid = i \qquad 1 \leq ctid' \leq n}
    {\smallstep{\VariableValuation,   ctid, (\Prog_1,
      \ldots, \Yield, \ldots, \Prog_n)}{\VariableValuation,  
    ctid', (\Prog_1, \ldots, \Skip, \ldots, \Prog_n)}}
  \end{display}
\end{figure}

\clearpage

\begin{figure}
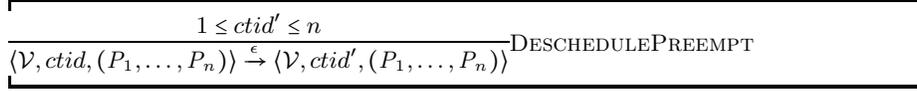

  \caption{From non-preemptive semantics to preemptive semantics\label{fig:preemptive_semantics}}
  \begin{display}{}
    \staterule{DeschedulePreempt}
    {1 \leq ctid' \leq n}
    {\smallstep{\VariableValuation,  
        ctid, (\Prog_1, \ldots, \Prog_n)}{\VariableValuation,  
        ctid', (\Prog_1, \ldots, \Prog_n)}}
  \end{display}
\end{figure}

\section{Proof of \thmref{correctness}}
\label{app:proof}

\correctness*

\begin{proof}
Let us assume $\sem{\cProg'}^{P}_{abs} \subseteq \sem{\cProg}^{NP}_{abs}$. 

Let $\sigma'$ be a concrete observation sequence in $\sem{\cProg'}^{P}$. 
Let $\sigma'_{abs}$ be the abstract observation sequence in $\sem{\cProg'}^{P}_{abs}$
corresponding to $\sigma'$. 
Then, there exists $\sigma_{abs} \in \sem{\cProg}^{NP}_{abs}$ 
such that $\sigma_{abs}$ is equivalent to $\sigma'_{abs}$.

Observe that if two abstract observation sequences --- 
$\sigma'_{abs}$ from $\sem{\cProg'}^{P}_{abs}$ and 
$\sigma_{abs}$ from $\sem{\cProg}^{NP}_{abs}$ --- are equivalent, 
then they correspond to executions over the {\em same} observable control-flow paths 
with the same data-flow into havoc and input/output statements.
Hence, $\sigma'_{abs}$ and $\sigma_{abs}$ either both map back to 
infeasible concrete observation sequences, or both map back to 
feasible, {\em equivalent} concrete observation sequences.

Since $\sigma'_{abs}$ maps back to a feasible concrete observation sequence $\sigma'$ 
by definition, $\sigma_{abs}$ also maps back to a feasible concrete observation sequence, 
say  $\sigma$, such that $\sigma$ is equivalent to $\sigma'$.
Hence, we have $\sem{\cProg'}^{P} \subseteq \sem{\cProg}^{NP}$.
\qed
\end{proof}

\clearpage

\section{Language Inclusion Procedure}
\label{app:algo}

The algorithm for $k$-bounded language inclusion \upto~$I$ is presented as
function \textsc{Inclusion} in \algoref{inclusionnew} (ignore Lines
\ref{line:overflowcheck}-\ref{line:markdirty} for now) .  
The function proceeds exactly
as the standard antichain algorithm outlined earlier.  
It explores $\A$ nondeterministically as before, and $\autCB$ is determinized on the fly for
exploration.  
The antichain and frontier sets consist of tuples of the form
$(s_\A, \autC)$, where $s_\A \in Q_\A$ and $\autC \subseteq Q_\B \times
\Sigma^k \times \Sigma^k$.  
Each tuple in the frontier set is first checked for
equivalence w.r.t. acceptance (\lineref{acceptcheck}).  
If this check fails, the function reports language inclusion failure (\lineref{fail}).  
If this check succeeds, the  successors are computed (\lineref{succ}). If a successor satisfies Rule 1,
it is ignored (\lineref{antichain1}), otherwise it is added to the frontier (\lineref{updfront})
and the antichain (\lineref{updantichain}). During the update of the antichain the algorithm ensures
that its invariant is preserved according to rule 2. The frontier also stores a sequence of symbols
that lead to a particular tuple of states in order to return a counterexample trace if
language inclusion fails.

We develop a procedure to check 
language inclusion \upto~$I$ by iteratively increasing the bound $k$ (see \algoref{inclusionnew} in the appendix). The procedure is {\em incremental}: the check for $k+1$-bounded language inclusion \upto~$I$ only explores paths along which the bound $k$ was exceeded in the previous iteration. 
Given a newly computed successor $(s_\A', \autC')$ for an iteration with bound $k$, 
if there exists some $(s_\B,\eta_1,\eta_2)$ in $\autC'$ such that the length of 
$\eta_1$ or $\eta_2$ exceeds $k$ (\lineref{overflowcheck}), 
we remember the tuple $(s_\A', \autC')$ in the set $\overflow$ (\lineref{updoverflow}).
We continue exploration of $\autCB$ from all states $(s_\B,\eta_1,\eta_2)$ 
with $|\eta_1| \leq k \wedge |\eta_2| \leq k$, but mark them \dirty. 
If we find a counter-example to language inclusion we return it and test
if it is spurious (\lineref{testsp}).
It may be a spurious counter-example caused because we removed states exceeding $k$.
In that case we increase the bound to $k+1$, remove all dirty items from the antichain
and frontier (lines \ref{line:cleanfrontier}-\ref{line:cleanantichain}), and
add the items from the overflow (\lineref{overflow}).
Intuitively this will undo all exploration from the point(s) the bound was 
exceeded and restarts from that/those point(s).

To test if a particular counterexample is spurious,
we invoke the language inclusion procedure,
replacing the preemptive automaton with the exact 
trace (trace automaton) and allowing an infinite bound.  This is 
fast and guaranteed to terminate as the trace automaton does 
not have loops.  We found that this optimization helps find a 
valid counterexample faster.

\clearpage

\begin{algorithm}[H]
\begin{algorithmic}[1]
 \Require Automata $A=(Q_A,\Sigma_A,\Delta_A,I_A,F_A)$ and $B=(Q_B,\Sigma_B,\Delta_B,I_B,F_B)$
 \Ensure $\true$ only if $\Lang{A} \subseteq \CI(\Lang B)$, $\false$ only if $\Lang{A} \not\subseteq \CI(\Lang B)$
  \State $\frontier\gets\{(s_\A,\{(I_\B,\emptyset,\emptyset)\},\emptyset): s_\A \in I_\A\}$   
  \State All tuples in $\frontier$ are not \dirty
  \State $\antichain\gets\frontier$
  \State $\overflow\gets\emptyset$
  \State $k\gets 2$
  \While{$\true$}
    \State $cex\gets\Call{inclusion}{k}$
    \If{$cex\neq\true \wedge cex$ is spurious} \label{line:testsp}
      \State $k\gets k+1$
      \State $\frontier\gets\{(s_\A,\autC)\in\frontier: \autC\mbox{ not \dirty}\}\cup\overflow$     \label{line:cleanfrontier}
      \State $\antichain\gets\{(s_\A,\autC)\in\antichain: \autC\mbox{ not \dirty}\}\cup\overflow$   \label{line:cleanantichain}
      \State $\overflow\gets\emptyset$                                                          \label{line:overflow}
    \Else
      \State\Return $cex$
    \EndIf
  \EndWhile
  \Statex

 \Function{inclusion}{$k$}
  \While {$\frontier\neq\emptyset$}
    \State remove a tuple $(s_\A,\autC,cex)$ from $\frontier$
    \If{$s_\A \in F_\A$ $\wedge$ $(\autC \cap F_\B) = \emptyset$} \label{line:acceptcheck}
      \Return $cex$                                     \label{line:fail}
    \EndIf
    \ForAll{$\alpha\in\Sigma$}
      \State $(s'_\A,\autC') \gets \suc_\alpha(s_\A,\autC)$ \label{line:succ}
       \If{$\nexists p \in\antichain: p \sqsubseteq (s'_\A, \autC')$} \label{line:antichain1} \Comment{\text{Rule} 1} 
          \If{$\exists(s_\B,\eta_1,\eta_2)\in \autC':\ |\eta_1|>k \, \vee \, |\eta_2| > k$} \label{line:overflowcheck}
            \If{$\autC'$ not \dirty} $\overflow\gets\overflow\cup\{(s_\A',\autC')\}$ \EndIf \label{line:updoverflow}
            \State $\autC' \gets \{(s_\B,\eta_1,\eta_2) \in \autC': \ |\eta_1| \leq k \, \wedge \, |\eta_2| \leq k\}$  \label{line:removerflow}
            \State Mark $\autC'$ \dirty           \label{line:markdirty}
          \EndIf
          \State $\frontier\gets\frontier \cup \{(s_\A',\autC',cex\cdot\alpha)\}$ \label{line:updfront}
          \State $\antichain\gets\antichain \backslash \{p: \autC' \sqsubseteq p\} \cup \{(s_\A',\autC')\}$ \Comment{\text{Rule} 2} \label{line:updantichain}
        \EndIf
    \EndFor
  \EndWhile
  \State\Return $\true$
 \EndFunction
\end{algorithmic}
\caption{Checking language inclusion \upto~$I$}
\label{algo:inclusionnew}
\end{algorithm}

\clearpage

\section{Synchronization inference rules}
\label{app:patterns}

The inference rules are applied as rewrite rules to the formula 
$\rho_g$ obtained in \secref{algo}. Each rule requires a 
certain subexpression in $\rho_g$ and rewrites it to a synchronization primitive.
That means that a single $\rho_g$ could possibly be solved by one of several synchronization
primitives.

The two lock rules fix atomicity violations and the reorder rule
fixes ordering violations.
The {\sc Add.Lock} rule captures a set of threads where thread 1 is descheduled at or after location $\loc_1$
and thread 2 is scheduled at or before $\loc_2$. Another context switch deschedules thread 2 at or after $\loc'_2$
and schedules again thread 1 at or before $\loc'_1$. As this pattern is present in the generalized $\rho_g$
this context switch is necessary to make the trace bad. We can avoid this context switch by adding the lock
from the conclusion. The {\sc Add.Lock2} rules captures the more general case where both, thread 2 interrupting
thread 1 and thread 1 interrupting thread 2, are bad traces.

The {\sc Add.Reorder} rule captures an ordering violation that can be fixed by 
moving a \Signal{} statement. Intuitively the \Await{} statement is
signaled too early and thread 1 can start running in the preemptive semantics.
In the non-preemptive semantics thread 2 keeps running after a \Signal{} statement until 
a preemption point is reached.

\begin{figure}[H]
 \vspace{10pt}
 \begin{minipage}{0.99\linewidth}
 \footnotesize{
  \centering
  \begin{tabular}[t]{@{}c@{\;\;}c@{\;\;}c}
    \infer[\textsc{Add.Lock}]{\mtt{lock}(\tid_1.[\loc_1:\loc'_1],\tid_2.[\loc_2:\loc'_2]) \vee \mixcons}
    {\rho_g = \tid_1.\loc_1 < \tid_2.\loc'_2  \; \wedge \;
\tid_2.\loc_2 < \tid_1.\loc'_1
       \;\wedge\; \mixcons 
    }\\[\GAP]
    \infer[\textsc{Add.Lock2}]{\mtt{lock}(\tid_1.[\loc_1:\loc'_1],\tid_2.[\loc_2:\loc'_2]) \vee \mixcons}
    {\rho_g = \tid_1.\loc_1 < \tid_2.\loc_2  \; \wedge \;
\tid_2.\loc'_2 < \tid_1.\loc'_1
       \;\wedge\; \mixcons 
    }\\[\GAP]
    \infer[\textsc{Add.Reorder}]{\mtt{reorder}(\tid_2.\loc_2,\tid_2.\loc'_2) \vee \mixcons}
    {\rho_g = \tid_1.\loc'_1 < \tid_2.\loc'_2
       \; \wedge \; \mixcons 
      & \exists \tid_1.\loc_1,\tid_2.\loc_2:\; 
      & \tid_1.\loc_1 < \tid_1.\loc'_1 \\
      \tid_2.\loc_2 < \tid_2.\loc'_2
      & \tid_1.\loc_1 = \Await{c}
      & \tid_2.\loc_2 = \Signal{c}
    }
  \end{tabular}
  }
 \end{minipage}
 \caption{Synchronization inference rules}
 \label{fig:fix_patterns}
\end{figure}
\clearpage

\end{appendix}

\end{document}